\documentclass[runningheads]{llncs}

\usepackage[utf8]{inputenc}
\usepackage{graphicx}
\graphicspath{{./Figures/}}

\usepackage{subcaption}
\usepackage{thm-restate}

\usepackage{xcolor}
\usepackage{amssymb,amsfonts, amsmath}
\usepackage{xspace,soul,wrapfig}
\usepackage{paralist,enumerate,cite}
\usepackage[textsize=footnotesize]{todonotes}
\usepackage[full]{complexity}

\usepackage{hyperref}
\usepackage[capitalise]{cleveref}

\newcommand{\fvig}{face-vertex incidence graph\xspace}
\newcommand{\cfc}{connected face cover\xspace}
\newcommand{\ops}{\textsc{Outerplane Splitting Number}\xspace}
\newcommand{\VC}{\textsc{Vertex Cover}\xspace}
\newcommand{\CFC}{\textsc{Connected Face Cover}\xspace}

\renewcommand{\orcidID}[1]{\href{https://orcid.org/#1}{\includegraphics[scale=.03]{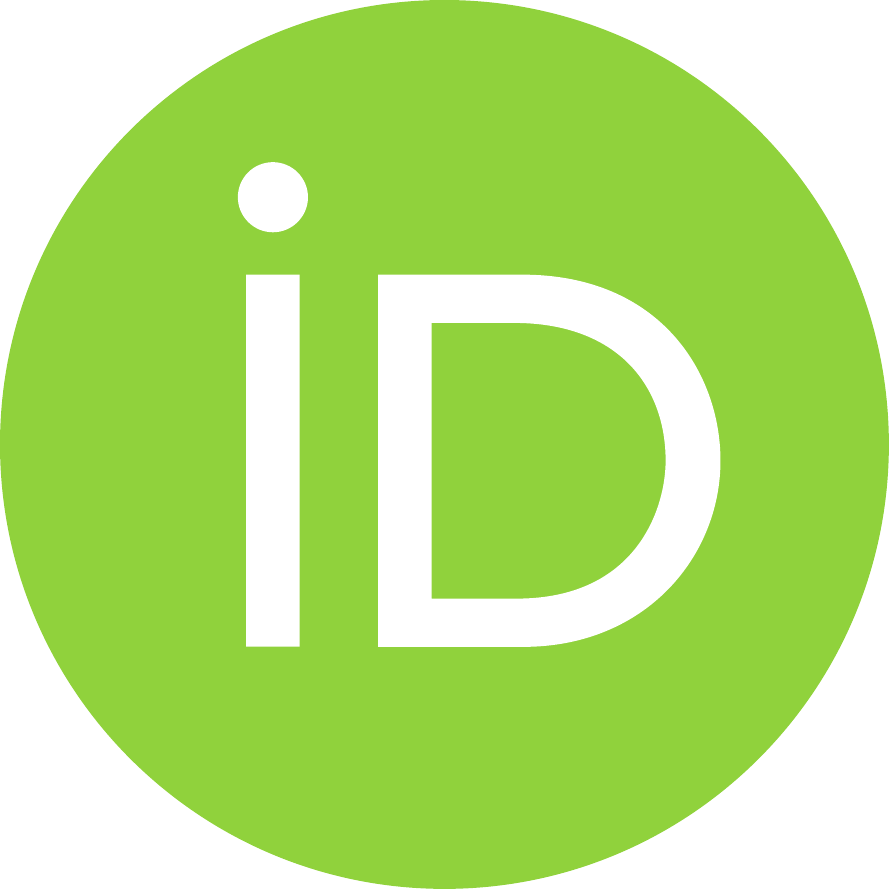}}}

\newcommand{\osn}{\ensuremath{\operatorname{osn}}}

\spnewtheorem{prp}{Property}{\bfseries}{\itshape}
\Crefname{prp}{Property}{Properties}

\spnewtheorem{observation}{Observation}{\bfseries}{\itshape}
\Crefname{observation}{Observation}{Observations}

\begin{document}
\title{Splitting Plane Graphs to Outerplanarity\thanks{Anaïs
Villedieu is supported by the Austrian Science Fund (FWF) under grant P31119.}}

\author{Martin Gronemann\orcidID{0000-1111-2222-3333} \and
	Martin Nöllenburg\orcidID{1111-2222-3333-4444} \and
	Anaïs Villedieu\orcidID{2222--3333-4444-5555}}

\institute{Algorithms and Complexity Group, TU Wien, Vienna, Austria
\email{\{mgronemann|noellenburg|avilledieu\}@ac.tuwien.ac.at}}

\maketitle


%
%

\begin{abstract}
Vertex splitting replaces a vertex by two copies and partitions its incident edges amongst the copies. 
This problem has been studied as a graph editing operation to achieve desired properties with as few splits as possible, most often planarity, for which the problem is \NP-hard.
Here we study 
how to minimize the number of splits to turn a plane graph into an outerplane one. We tackle this problem by establishing a direct connection between splitting a plane graph to outerplanarity, finding a connected face cover, and finding a feedback vertex set in its dual. We prove \NP-completeness for  plane biconnected graphs, while we show that a  polynomial-time algorithm exists for maximal planar graphs. Finally, we provide upper and lower bounds for certain families of maximal planar graphs.

\end{abstract}

\section{Introduction}

Graph editing problems are fundamental problems in graph theory. They define a set of basic operations on a graph $G$ and ask for the minimum number of these operations necessary in order to turn $G$ into a graph of a desired target graph class $\mathcal G$~\cite{nss-ccsemp-01,ly-nphpn-80,y-nenp-78,k-aog-96}. 
For instance, in the Cluster Editing problem~\cite{sst-cgmp-04} the operations are insertions or deletions of individual edges and the target graph class are cluster graphs, i.e., unions of vertex-disjoint cliques. 
In graph drawing, a particularly interesting graph class are planar graphs,  
for which several related graph editing problems have been studied, e.g., how many vertex deletions are needed to turn an arbitrary graph into a planar one~\cite{ms-opgvd-12} or how many vertex splits are needed to obtain a planar graph~\cite{jackson1984splitting,Faria_2001}. 
In this paper, we are interested in the latter operation: vertex splitting.
A \emph{vertex split} creates two copies of a vertex $v$, distributes its edges among these two copies and then deletes $v$ from~$G$. 

Further, we are translating the graph editing problem into a more geometric or topological drawing editing problem.
This means that we apply the splitting operations not to the vertices of an abstract graph, but to the vertices of a planar graph drawing, or more generally to a planar embedded (or \emph{plane}) graph. 
In a plane graph, each vertex has an induced cyclic order of incident edges, which needs to be respected by any vertex split in the sense that we must split its cyclic order into two contiguous intervals, one for each of the two copies. 
From a different perspective, the two faces that serve as the separators of these two edge intervals are actually merged into a single face by the vertex split.

Finally, we consider outerplanar graphs as the target graph class. 
Thus, we want to apply a minimum number of vertex splits to a plane graph $G$, which merge a minimum number of faces in order to obtain an outerplanar embedded graph $G'$, where all vertices are incident to a single face, called the \emph{outer face}. We denote this minimum number of splits as the \emph{outerplane splitting number} $\osn(G)$ of $G$ (see~\cref{fig:introfig}).
Outerplanar graphs are a prominent graph class in graph drawing (see, e.g., \cite{Biedl11,Frati22,LazardLL19,LenhartL96prox}) as well as in graph theory and graph algorithms more generally (e.g., \cite{BodlaenderF02,f-saicrt-96,mz-emaog-99}). 
For instance, outerplanar graphs admit planar circular layouts or 1-page book embeddings~\cite{bk-btg-79}. 
Additionally, outerplanar graphs often serve as a simpler subclass of planar graphs with good algorithmic properties. For instance, they have treewidth 2 and
their generalizations to $k$-outerplanar graphs still have bounded treewidth~\cite{Bodlaender98,Biedl_2015}, which allows for polynomial-time algorithms for \NP-complete problems that are tractable for such bounded-treewidth graphs. This, in turn, can be used to obtain a PTAS for these problems on planar graphs~\cite{Baker94}. 


We are now ready to define our main computational problem as follows.

\begin{problem}[\ops]
	Given a plane biconnected graph $G=(V,E)$ and an integer $k$, can we transform $G$ into an outerplane graph $G'$ by applying at most $k$ vertex splits to $G$?
\end{problem}

\begin{figure}[t]
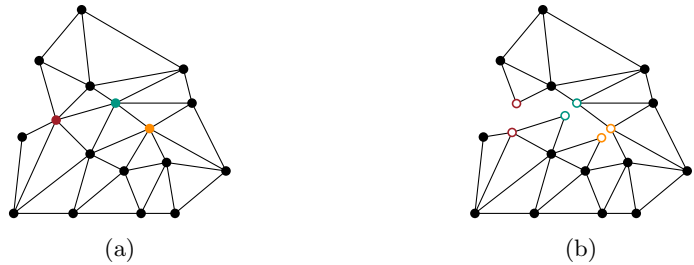

\centering
\begin{subfigure}[b]{.49\textwidth}
    \centering
    \includegraphics[page=3]{intro.pdf}
    \caption{}
    \label{fig:subfig3}
\end{subfigure}
\hfil
\begin{subfigure}[b]{.49\textwidth}
    \centering
    \includegraphics[page=4]{intro.pdf}
    \caption{}
    \label{fig:subfig3}
\end{subfigure}
\caption{(a)~An instance of \ops, where the colored vertices will be split; (b) resulting outerplane graph after the minimum 3 splits.} 
\label{fig:introfig}
\end{figure}

\paragraph{Contributions.}
In this paper, we introduce the above  problem \textsc{Outerplane Splitting Number. We start by showing} the key property for our subsequent results, namely that (minimum) sets of vertex splits to turn a  plane biconnected graph $G$ into an outerplane one correspond to (minimum) connected face covers in $G$ (\cref{sec:fvig}), which in turn are equivalent to (minimum) feedback vertex sets in the dual graph of $G$.
Using this tool we then show that for general  plane biconnected graphs \ops is \NP-complete (\cref{sec:hardness}), whereas for maximal planar graphs we can solve it in polynomial time (\cref{sec:fvs}). 
Finally, we provide upper and lower bounds on the outerplane splitting number for maximal planar graphs (\cref{sec:bounds}). 

\paragraph{Related Work.}
Splitting numbers have been studied mostly for abstract (non-planar) graphs with the goal of turning them into planar graphs. 
The \textsc{Planar Splitting Number} problem is \NP-complete in general~\cite{Faria_2001}, but exact splitting numbers are known for complete and complete bipartite graphs~\cite{HartsfieldJR85,jackson1984splitting}, as well as for the 4-cube~\cite{FariaFN98}. For two-layer drawings of general bipartite graphs, the problem is still \NP-complete, but \FPT~\cite{splittingbipfpt} when parametrized by the number of split vertices. 
It has also been studied for other surfaces such as the torus~\cite{Hartsfield86} and the projective plane~\cite{Hartsfield87}. 
Another related concept is the split thickness of a graph $G$ (or its folded covering number~\cite{ku-twcg-16}), which is the smallest $k$ such that $G$ can be transformed into a planar graph by applying at most $k$ splits per vertex. 
Recognizing graphs with split thickness $2$ is \NP-hard, but there is a constant-factor approximation algorithm and a fixed-parameter algorithm for graphs of bounded treewidth~\cite{Eppstein_2017}.
Recently, the complexity of the embedded splitting number problem of transforming non-planar graph drawings into plane ones has been investigated~\cite{splittingnoellenburg}.
Beyond the theoretical investigations of splitting numbers and planarity, there are also applied work in graph drawing making use of vertex splitting to untangle edges~\cite{Wu_2020_TVCG} or to improve layout quality for community exploration~\cite{admpt-hgvwcaea-22,HenryBF08}.

Regarding vertex splitting for achieving graph properties other than planarity, Trotter and Harary~\cite{th-dmig-79} studied vertex splitting to turn a graph into an interval graph. 
Paik et al.~\cite{PaikRS98} considered vertex splitting to remove long paths in directed acyclic graphs and
Abu-Khzam et al.~\cite{Abu-KhzamBFS21} studied heuristics using vertex splitting for a cluster editing problem.




\begin{figure}[t]
\centering
\begin{subfigure}[b]{.24\textwidth}
    \centering
    \includegraphics[page=1]{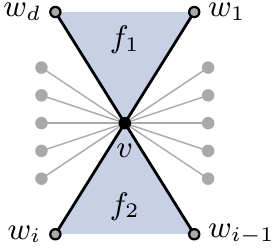}
    \caption{}
    \label{fig:split_intro_1}
\end{subfigure}
\hfil
\begin{subfigure}[b]{.24\textwidth}
    \centering
    \includegraphics[page=2]{split_intro}
    \caption{}
    \label{fig:split_intro_2}
\end{subfigure}
\begin{subfigure}[b]{.24\textwidth}
    \centering
    \includegraphics[page=3]{split_intro}
    \caption{}
    \label{fig:split_intro_3}
\end{subfigure}
\begin{subfigure}[b]{.24\textwidth}
    \centering
    \includegraphics[page=4]{split_intro}
    \caption{}
    \label{fig:split_intro_4}
\end{subfigure}
\caption{(a)~Two touching faces $f_1, f_2$ with a common vertex $v$ on their boundary. (b)~Result of the split of $v$ with respect to $f_1, f_2$ joining them into a new face $f$. (c-d)~Merging 4 faces $f_1, \ldots, f_4$ covering a single vertex $v$ with 3 splits.}
	\label{fig:split_intro}
\end{figure}

\paragraph{Preliminaries.} The key concept of our approach is to merge a set of faces of a given plane graph $G=(V,E)$ with vertex set $V=V(G)$ and edge set $E=E(G)$ into one big face which is incident to all vertices of $G$. Hence, the result is outerplanar.
The idea is that if two faces $f_1$ and $f_2$ share a vertex $v$ on their boundary (we say $f_1$ and $f_2$ \emph{touch}, see \cref{fig:split_intro_1}), then we can split $v$ into two new vertices $v_1, v_2$. In this way, we are able to create a narrow gap, which merges $f_1, f_2$ into a bigger face $f$ (see \cref{fig:split_intro_2}). With this in mind, we formally define an \emph{embedding-preserving split} of a vertex $v$ w.r.t. two incident faces $f_1$ and $f_2$. We construct a new plane graph $G'=(V', E')$ with $V' = V \setminus \{v\} \cup \{v_1,v_2\}$. Consider the two neighbors of $v$ both incident to $f_1$ and 
let $w_1$ be the second neighbor in clockwise order. 
Similarly, let $w_i$ be the second vertex adjacent to $v$ and incident to $f_2$. We call $w_d$ the vertex preceding $w_1$ in the cyclic ordering or the neighbors, with $d$ being the degree of $v$, see~\cref{fig:split_intro_1}. Note that while $w_1=w_{i-1}$ and $w_i=w_d$ is possible, $w_d\ne w_1$ and $w_{i-1}\ne w_i$.
For the set of edges, we now set $E'= E \setminus \{ (v,w_1), \ldots, (v,w_d) \} \cup \{ (v_2, w_1), \ldots, (v_2, w_{i-1}) \} \cup \{ (v_1, w_i), \ldots, (v_1, w_{d}) \}$ and assume that they inherit their embedding from $G$.
From now on we refer to this operation simply as a \emph{split} or when $f_1,f_2$ are clear from the context, we may refer to \emph{merging} the two faces at $v$. 
The vertices $v_1,v_2$ introduced in place of $v$ are called \emph{copies} of $v$. 
If a copy $v_i$ of a vertex $v$ is split again, then any copy of $v_i$ is also called a copy of the original vertex $v$. 

We can now reformulate the task of using as few splits as possible. 
Our objective is to find a set of faces $S$ that satisfies two conditions.
(1) Every vertex in $G$ has to be on the boundary of at least one face $f \in S$, that is, the faces in $S$ \emph{cover} all vertices in $V$.\footnote{Testing whether such $S$ with $|S| \le k$ exists, is the \NP-complete  problem \textsc{Face Cover}~\cite{bm-ccvfpg-88}.}
And (2) for every two faces $f,f' \in S$ there exists a set of faces $\{f_1,\dots,f_k\}\subseteq S$ 
such that $f=f_1, \ldots, f_k = f'$, 
and $f_i$ touches $f_{i+1}$ for $1\leq i < k$. 
In other words, $S$ is connected in terms of touching faces. We now introduce the main tool in our constructions that formalizes this concept. 

\section{Face-Vertex Incidence Graph}\label{sec:fvig}
Let $G=(V,E)$ be a plane biconnected graph and $F$ its set of faces. The \emph{\fvig} is defined as $H = (V \cup F, E_H)$ and contains the edges $E_H = \{ (v,f) \in V \times F : v \text{ is on the boundary of } f \}$. 
Graph $H$ is by construction bipartite and we  
assume that it is plane by placing each vertex $f \in F$ into its corresponding face in~$G$. 


\begin{definition}\label{def:cfc} 
Let $G$ be a plane biconnected graph, let $F$ be the set of faces of $G$, and let $H$ be its \fvig. 
A \emph{face cover} of $G$ is a set $S \subseteq F$ of faces such that every vertex $v\in V$ is incident to at least one face in $S$.
A face cover $S$ of 
$G$ is a \emph{connected} face cover if the induced subgraph $H[S \cup V]$ of $S \cup V$ in $H$ is connected.
\end{definition}
\noindent 
We point out that the problem of finding a \cfc is not equivalent to the Connected Face Hitting Set Problem \cite{SchweitzerS10}, where a connected set of vertices incident to every face is computed. 
\noindent We continue with two lemmas that are concerned with 
merging multiple faces at the same vertex (\cref{fig:split_intro_3}). 

\begin{lemma}\label{lem:single-vertex-split}
Let $G$ be a plane biconnected graph and $S \subseteq F$ a subset of the faces $F$ of $G$ that all have the vertex $v \in V$ on their boundary. Then $|S|-1$ splits are sufficient to merge the faces of $S$ into one.
\end{lemma}

\begin{proof}
Let $f_1, \ldots, f_k$ with $k = |S|$ be the faces of $S$ in the clockwise order as they appear around $v$ ($f_1$ chosen arbitrarily).
We iteratively merge $f_1$ with $f_i$ for $2 \leq i \leq k$, which requires in total $|S|-1$ splits (see \cref{fig:split_intro_3} and \cref{fig:split_intro_4}). \qed
\end{proof}

\begin{lemma}\label{lem:cover-to-split}
Let $G$ be a plane biconnected graph and let $S$ be a \cfc of $G$. Then $|S|-1$ splits are sufficient to merge the faces of $S$ into one. 
\end{lemma}

\begin{proof}
Let $H' = H[S \cup V]$ and
compute a spanning tree $T$ in $H'$. For every vertex $v \in V(T) \cap V(G)$, we apply \cref{lem:single-vertex-split} with the face set $F'(v) = \{ f \in S \cap V(T) \; | \; (v,f)  \in E(T) \}$. 
We root the tree at an arbitrary face $f' \in S$, which provides a hierarchy on the vertices and faces in $T$.  Every vertex $v \in V(T) \cap V(G)$ requires by \cref{lem:single-vertex-split} $|F'(v)|-1$ splits. Note that that for all leaf vertices in $T$, $|F'(v)|=1$, i.e., they will not be split. Each split is charged to the children of $v$ in $T$. Since $H$ is bipartite, so is $T$. It follows that every face $f \in S \setminus \{ f' \}$ is charged exactly once by its parent, thus $|S|-1$ splits suffice. \qed
\end{proof}

\begin{lemma}\label{lem:split-to-cover}
Let $G$ be a plane biconnected graph and $\sigma$ a sequence of $k$ 
splits to make $G$ outerplane. Then $G$ has a \cfc of size $k+1$.
\end{lemma}

\begin{proof}
    Since by definition applying $\sigma$ to $G$ creates a single big face that is incident to all vertices in $V(G)$ by iteratively merging pairs of original faces defining a set $S\subseteq F$, it is clear that $S$ is a face cover of $G$ and since the result of the vertex splits and face merges creates a single face, set $S$ must also be connected. \qed
\end{proof}

As a consequence of \cref{lem:cover-to-split,lem:split-to-cover} we obtain that \ops and computing a minimum \cfc are equivalent. 

\begin{theorem}\label{thm:split-cfc}
    Let $G$ be a plane biconnected graph. Then $G$ has outerplane splitting number $k$ if and only if it has a connected face cover of size $k+1$.
\end{theorem}


\section{\NP-completeness}\label{sec:hardness}
In this section, we prove that finding a \cfc of size $k$ (and thus \ops) is \NP-complete. The idea is to take the dual of a planar biconnected \VC instance and subdivide every edge once (we call this an \emph{all-1-subdivision}). Note that the all-1-subdivision of a graph $G$ corresponds to its vertex-edge incidence graph and the all-1-subdivision of the dual of $G$ corresponds to the face-edge incidence graph of $G$. A \cfc then corresponds to a vertex cover in the original graph, and vice versa. The following property greatly simplifies the arguments regarding \cref{def:cfc}. 

\begin{prp}\label{prp:subdiv}
	Let $G'$ be an all-1-subdivision of a biconnected planar graph $G$ and $S$ a set of faces that cover $V(G')$. Then $S$ is a \cfc of $G'$.
\end{prp}


\begin{proof}
Let $H$ be the all-1-subdivision of the dual of $G$, and assume to the contrary that the induced subgraph $H' = H[S \cup V(G)]$ is not connected. Then there exists an edge $(u,v) \in E(G)$ such that $u$ and $v$ are in different connected components in $H'$. 
Let $w$ be the subdivision vertex of $(u,v)$ in $G'$. 
As a subdivision vertex, $w$  is incident to only two faces, one of which, say $f$, must be contained in $S$. 
But $f$ is also incident to $u$ and $v$ and hence $u$ and $v$ are in the same component of $H'$ via face $f$, a contradiction. Hence $H'$ is connected and $S$ is a \cfc of $G'$. \qed
\end{proof}

\begin{figure}[t]
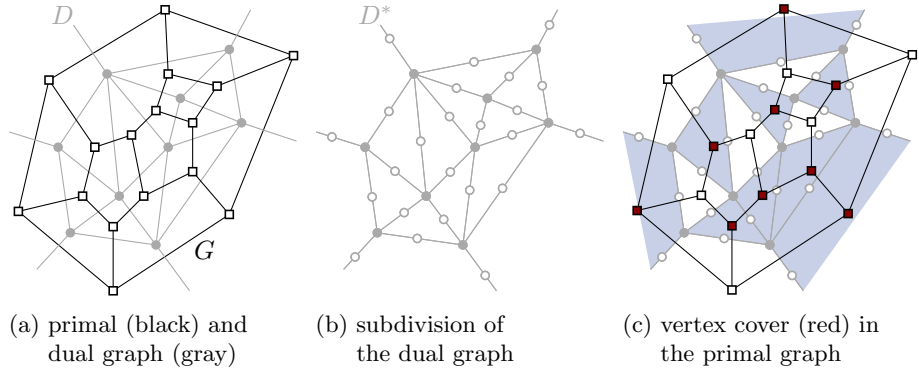

	\centering
	\begin{subfigure}[b]{.32\textwidth}
		\centering
		\includegraphics[page=2]{VC_reduction.pdf}
		\caption{primal (black) and\\dual graph (gray)}
		\label{fig:subfig99}
	\end{subfigure}
	\hfil
	\begin{subfigure}[b]{.32\textwidth}
		\centering
		\includegraphics[page=3]{VC_reduction.pdf}
		\caption{subdivision of\\the dual graph}
		\label{fig:subfig99}
	\end{subfigure}
	\hfil
	\begin{subfigure}[b]{.32\textwidth}
		\centering
		\includegraphics[page=4]{VC_reduction.pdf}
		\caption{vertex cover (red) in\\the primal graph}
		\label{fig:subfig99}
	\end{subfigure}
	\caption{
	Link between the primal graph $G$, its vertex cover, the dual $D$ and its subdivision $D^*$.}
	\label{fig:vertex_cover}
\end{figure}

The proof of the next theorem is very similar to the reduction of Bienstock and Monma to show \NP-completeness of \textsc{Face Cover}~\cite{bm-ccvfpg-88}; due to differences in the problem definitions, such as the connectivity of the face cover and whether the input graph is plane or not, we provide the full reduction for the sake of completeness.

\begin{theorem}
Deciding whether a plane biconnected graph $G$ has a \cfc of size at most $k$ is \NP-complete.
\end{theorem}


\begin{proof}
Clearly the problem is in $\NP$. To prove hardness, we first introduce some notation. Let $G$ be a plane biconnected graph and $D$ the corresponding dual graph.
Furthermore, let $D^*$ be 
the all-1-subdivision of $D$.
We prove now that a \cfc $S^*$ of size $k$ in $D^*$ is in a one-to-one correspondence with a vertex cover $S$ of size $k$ in $G$ (see \cref{fig:vertex_cover}). More specifically, we show that the dual vertices of the faces of $S^*$ that form a \cfc in $D^*$, are a vertex cover for $G$ and vice versa. The reduction is from the \NP-complete \VC problem in biconnected 
planar graphs in which all vertices have degree 3 (cubic graphs)~\cite{Mohar01}.

\paragraph{\CFC $\Rightarrow$ \VC:} 
Let $G$ be such a biconnected plane \VC instance. Assume we have a \cfc $S^*$ with $|S^*|=k$ for $D^*$. Note that the faces of $D^*$ correspond to the vertices in $G$. We claim that the faces $S^*$, when mapped to the corresponding vertices $S \subseteq V(G)$ are a vertex cover for $G$. Assume otherwise, that is, there exists an edge $e^* \in E(G)$ that has no endpoints in $S$. However, $e^*$ has a dual edge $e \in E(D)$ and therefore a subdivision vertex $v_e \in V(D^*)$. Hence, there is a face $f \in S^*$ that has $v_e$ on its boundary by definition of \cfc. And when mapped to $D$, $f$ has $e$ on its boundary, which implies that the primal edge $e^*$ has at least one endpoint in $S^*$; a contradiction.

\paragraph{\VC $\Rightarrow$ \CFC:} 
To prove that a vertex cover $S$ induces a \cfc $S^*$ in $D^*$, we have to prove that $S^*$ covers all vertices and the induced subgraph in the \fvig $H$ is connected. We proceed as in the other direction. $S$ covers all edges in $E(G)$, thus every edge $e \in E(D)$ is bounded by at least one face of $S^*$. Hence, every subdivision vertex in $V(D^*)$ is covered by a face of $S^*$. Furthermore, every vertex in $D^*$ is adjacent to a subdivision vertex, thus, also covered by a face in $S^*$.
Since $S^*$ is covering all vertices, we obtain from \cref{prp:subdiv} that $S^*$ is a \cfc.
\qed
\end{proof}


\section{Feedback Vertex Set Approach}\label{sec:fvs}


A \emph{feedback vertex set} $S^\circ\subset V(G)$ of a graph $G$ is a vertex subset such that the induced subgraph $G[V(G)\setminus S^\circ]$ is acyclic. 
We show here that finding a \cfc $S$ of size $k$ for a plane biconnected graph $G$ is equivalent to finding a feedback vertex set $S^\circ \subset V(D)$ of size $k$ in the dual graph $D$ of $G$. 
The \emph{weak dual}, i.e., the dual without a vertex for the outer face, of an outerplanar graph is a forest. Thus we must find the smallest number of splits in $G$ which transform $D$ into a forest.
In other words, we must must break all the cycles in $D$, and hence all of the vertices in the feedback vertex set $S^\circ$ of $D$ correspond to the faces of $G$ that should be merged together (see~\cref{fig:algsteps}).

\begin{figure}[t]
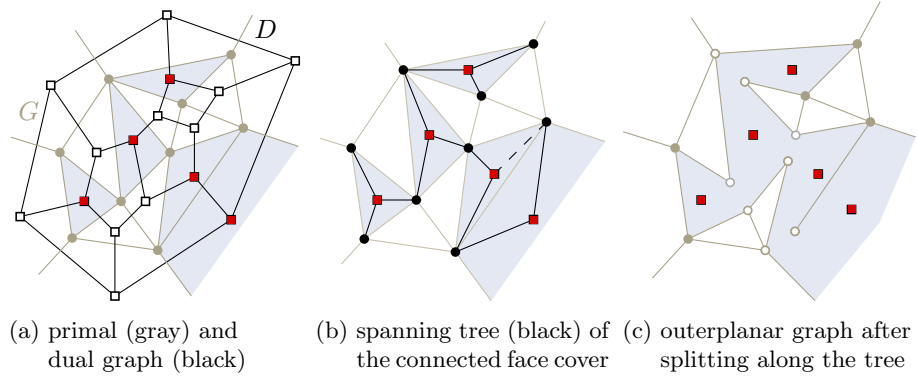

\centering
\begin{subfigure}[b]{.32\textwidth}
    \centering
    \includegraphics[page=2]{polyalg_steps_alt.pdf}
    \caption{primal (gray) and\\dual graph (black)}
    \label{fig:subfig3}
\end{subfigure}
\hfil
\begin{subfigure}[b]{.32\textwidth}
    \centering
    \includegraphics[page=3]{polyalg_steps_alt.pdf}
    \caption{spanning tree (black) of\\the \cfc}
    \label{fig:subfig3}
\end{subfigure}
\hfil
\begin{subfigure}[b]{.32\textwidth}
    \centering
    \includegraphics[page=4]{polyalg_steps_alt.pdf}
    \caption{outerplanar graph after\\splitting along the tree}
    \label{fig:subfig3}
\end{subfigure}
	\caption{
	The \cfc\ (blue) is a feedback vertex set (red) in the dual.
	}
	\label{fig:algsteps}
\end{figure}


\begin{prp}\label{prp:fvsconnec}
	Let $H$ be the \fvig of a plane biconnected graph $G$ and let $S^\circ$ be a feedback vertex set in the dual $D$ of $G$. 
	Then $S^\circ$ induces a \cfc $S$ in $G$.
\end{prp}

\begin{proof}
We need to show that $S^\circ$ is a face cover and that it is connected. 
First, assume there is a vertex $v \in V(G)$ of degree $\deg(v) = d$ that is not incident to a vertex in $S^\circ$, i.e., a face of $G$.
Since $G$ is biconnected, $v$ is incident to $d$ faces $f_1, \dots, f_d$, none of which is contained in $S^\circ$.
But then $D[V(D) \setminus S^\circ]$ has a cycle $(f_1, \dots, f_d)$, a contradiction. 

Next, we define $\overline{S^\circ} = V(D)\setminus S^\circ$ as the complement of the feedback vertex set $S^\circ$ in $D$. 
Assume that $H[V\cup S^\circ]$ has at least two separate connected components $C_1, C_2$. 
Then there must exist a closed curve in the plane separating $C_1$ from $C_2$, which avoids the faces in $S^\circ$ and instead passes through a sequence $(f_1, \dots, f_\ell)$ of faces in $\overline{S^\circ}$, where each pair $(f_i, f_{i+1})$ for $i \in \{1, \dots, \ell-1\}$ as well as $(f_\ell, f_1)$ are adjacent in the dual $D$. 
Again this implies that there is a cycle in $D[V(D) \setminus S^\circ]$, a contradiction. 
Thus $S^\circ$ is a \cfc. \qed
\end{proof}

\begin{theorem}\label{thm:split-fvs}
    A plane biconnected graph $G$ has outerplane splitting number $k$ if and only if its dual $D$ has a minimum feedback vertex set of size $k+1$.
\end{theorem}

\begin{proof}
    Let $S^\circ$ be a minimum feedback vertex set of the dual $D$ of $G$ with cardinality $|S^\circ| = k+1$ and let $H$ be the \fvig of $G$. 
    We know from \cref{prp:fvsconnec} that $H' = H[V(G) \cup S^\circ]$ is connected and hence $S^\circ$ induces a \cfc $S$ with $|S| = k+1$. Then by \cref{lem:cover-to-split} $G$ has $\osn(G) \le k$.
    
    Let conversely $\sigma$ be a sequence of $k$ vertex splits that turn $G$ into an outerplane graph $G'$ and let $F$ be the set of faces of $G$. 
    By \cref{lem:split-to-cover} we obtain a \cfc $S$ of size $k+1$ consisting of all faces that are merged by $\sigma$.
    The complement $\overline{S} = F\setminus S$ consists of all faces of $G$ that are not merged by the splits in $\sigma$ and thus are the remaining (inner) faces of the outerplane graph $G'$.
    Since $G'$ is outerplane and biconnected, $\overline{S}$ is the vertex set of the weak dual of $G'$, which must be a tree. 
    Hence $S$ is a feedback vertex set in $D$ of size $k+1$ and the minimum feedback vertex set in $D$ has size at most $k+1$. \qed
\end{proof}        
    
	


Since all faces in a maximal planar graph are triangles, the maximum vertex degree of its dual is 3. 
Thus, we can apply the polynomial-time algorithm of Ueno et al.~\cite{UenoKG88} to this dual, which computes the minimum feedback vertex set in graphs of maximum degree 3 by reducing the instance to polynomial-solvable matroid parity problem instance, and obtain

\begin{corollary}
	We can solve \ops\ for maximal planar graphs in polynomial time.
\end{corollary}

Many other existing results for feedback vertex set  extend to \ops, e.g., it has a kernel of size $13k$~\cite{BonamyK16} and admits a \PTAS~\cite{DemaineH05}.

\section{Lower and Upper Bounds}\label{sec:bounds}

In this section we provide some upper and lower bounds on the outerplane splitting number in certain maximal planar graphs.

\subsection{Upper Bounds}

Based on the equivalence of \cref{thm:split-fvs} we obtain upper bounds on the outerplane splitting number from suitable upper bounds on the feedback vertex set problem, which has been studied for many graph classes, among them cubic graphs~\cite{BondyHS87}. 
Liu and Zhao~\cite{LiuZ96} showed that cubic graphs $G=(V,E)$ of girth at least four (resp., three) have a minimum feedback vertex set of size at most $\frac{|V|}{3}$ (resp., $\frac{3|V|}{8}$). 
Kelly and Liu~\cite{kl-msfvspgglf-17} showed that connected planar subcubic graphs of girth at least five have a minimum feedback vertex set of size at most $\frac{2|V|+2}{7}$.
Recall that the girth of a graph is the length of its shortest cycle.

\begin{proposition}\label{cor:upper}
	The outerplane splitting number of a maximal planar graph $G=(V,E)$ of minimum degree (i) 3, (ii) 4, and (iii) 5, respectively, is at most (i) $\frac{3 |V| - 10}{4}$, (ii) $\frac{2|V|-7}{3}$, and (iii) $\frac{4|V|-13}{7}$, respectively.
\end{proposition}

\begin{proof}
    Maximal planar graphs with $n = |V|$ vertices have $2n-4$ faces. 
    So the corresponding dual graphs have $2n-4$ vertices. 
    Moreover, since the degree of a vertex in $G$ corresponds to the length of a facial cycle in the dual, graphs with minimum vertex degree 3, 4, or 5 have duals with girth 3, 4, or 5, respectively.
    So if the minimum degree in $G$ is 3, we obtain an upper bound on the feedback vertex set of $(3n - 6)/4$; if the minimum degree is 4, the bound is $(2n-4)/3$; and if the minimum degree is 5, the bound is $(4n-6)/7$. 
    The claim then follows from \cref{thm:split-fvs}. \qed
\end{proof}



\subsection{Lower Bounds}


We first provide a generic lower bound for the outerplane splitting number of maximal planar graphs. 
Let $G$ be an $n$-vertex maximal planar graph with $2n-4$ faces. 
Each face is a triangle incident to three vertices. 
In a minimum-size connected face cover $S^*$, the first face covers three vertices. 
Due to the connectivity requirement, all other faces can add at most two newly covered vertices. 
Hence we need at least $\frac{n-1}{2}$ faces in any connected face cover. 
By \cref{thm:split-cfc} this implies that $\osn(G) \ge \frac{n-3}{2}$.

\begin{proposition}
Any maximal planar graph $G$ has outerplane splitting number at least $\frac{|V(G)| - 3}{2}$.
\end{proposition}

Next, towards a better bound, we define a family of maximal planar graphs $T_d = (V_d, E_d)$ of girth 3 for $d \ge 0$ that have outerplane splitting number at least $\frac{2 |V_d| - 8}{3}$.
The family are the complete planar 3-trees of depth $d$, which are defined recursively as follows. 
The graph $T_0$ is the 4-clique $K_4$.
To obtain $T_d$ from $T_{d-1}$ for $d \ge 1$ we subdivide each inner triangular face of $T_{d-1}$ into three triangles by inserting a new vertex and connecting it to the three vertices on the boundary of the face. 

\begin{proposition}
  The complete planar 3-tree $T_d$ of depth $d$ has outerplane splitting number at least $\frac{2 |V_d| - 8}{3}$.
\end{proposition}

\begin{proof}
Each $T_d$ is a maximal planar graph with $n_d = 3 + \sum_{i=0}^{d} 3^i = \frac{3^{d+1} + 5}{2}$ vertices.
All $3^d$ leaf-level vertices added into the triangular faces of $T_{d-1}$ in the last step of the construction have degree 3 and are incident to three exclusive faces, i.e., there is no face that covers more than one of these leaf-level vertices. 
This immediately implies that any face cover of $T_d$, connected or not, has size at least $3^d$.
From $n_d = \frac{3^{d+1} + 5}{2}$ we obtain $d = \log_3 \frac{2n_d-5}{3}$ and $3^d = \frac{2n_d-5}{3}$. 
\cref{thm:split-cfc} then implies that $\osn(T_d) \ge \frac{2n_d-8}{3}$. \qed
\end{proof}

\section{Open Problems}
We have introduced the \ops\ problem and established its complexity for plane biconnected graphs. 
The most important open question revolves around the embedding requirement. 
Splitting operations can be defined more loosely and allow for any new embedding and neighborhood of the split vertices. 
In general, it is also of interest to understand how the problem differs when the input graph does not have an embedding at all, as in the original splitting number problem.
Since \ops\ can be solved in polynomial time for maximal planar graphs but is hard for plane biconnected graphs, there is a complexity gap to be closed when faces of degree more than three are involved.
Vertex splitting in graph drawings has so far been studied to achieve planarity and outerplanarity. A natural extension is to study it for other graph classes or graph properties.

\bibliographystyle{splncs04}
\bibliography{references}

\end{document}